\documentclass[11pt,a4paper]{article}
\usepackage{authblk}
\usepackage[margin=2cm]{geometry}
\usepackage[utf8]{inputenc}
\usepackage{amsmath}
\usepackage{amssymb}
\usepackage{amsthm}
\usepackage{dsfont}
\usepackage{color}
\usepackage{amsmath}
\usepackage{amsfonts}
\usepackage{amssymb}
\usepackage{bbm}
\usepackage[procnumbered, linesnumbered,algoruled]{algorithm2e}
\usepackage[compact]{titlesec}

\newcommand{\comment}[1]{}

\newcommand{\greedy}{\textsc{Greedy}}
\newcommand{\greedyplus}{\textsc{Greedy+Singleton}}
\newcommand{\enum}{\textsc{EnumGreedy}}
\newcommand{\eps}{{\varepsilon}}
\newcommand{\aeps}{{\mu}}

\DeclareMathOperator*{\argmax}{arg\,max}

\begin{document}
\newtheorem{thm}{Theorem}[section]
\newtheorem{prop}[thm]{Proposition}
\newtheorem{assm}[thm]{Assumption}
\newtheorem{lem}[thm]{Lemma}

\newtheorem{obs}[thm]{Observation}
\newtheorem{cor}[thm]{Corollary}
 \newtheorem{lemma}[thm]{Lemma}
 \newtheorem{proposition}[thm]{Proposition}
 \newtheorem{claim}[thm]{Claim}
\newtheorem{defn}[thm]{Definition}
\newcommand{\ariel}[1]{{\color{red} (Ariel :#1)}}
\def \II   {{\mathcal I}}
\newcommand{\one}{\mathbbm{1}}

\title{
A Refined Analysis of 
Submodular Greedy
}
\author[1]{Ariel Kulik\thanks{\texttt{kulik@cs.technion.ac.il}}}
\author[1]{Roy Schwartz\thanks{\texttt{schwartz@cs.technion.ac.il}}}
\author[1]{Hadas Shachnai\thanks{\texttt{hadas@cs.technion.ac.il}}}
\affil[1]{Computer Science Department, Technion, Haifa 3200003, Israel}
\maketitle
\begin{abstract}
\comment{
We present a $(1-e^{-1})$-approximation algorithm for maximizing a monotone submodular function subject to a knapsack constraint
in $O(n^4)$ function value computations. This improves the best known running time of $O(n^5)$ required by a $(1-e^{-1})$-approximation algorithm for the problem due to [Sviridenko, 2004].

Sviridenko's algorithm exhaustively iterates over all subsets of at most three elements and extends each using a greedy approach.
We improve the running time via a refined analysis of the greedy approach which shows that 
exhaustively iterating over all subsets of at most two elements suffices for obtaining the same approximation ratio.
}

\comment{
The problem of maximizing a monotone submodular function subject to a knapsack constraint admits
 a tight  $(1-e^{-1})$-approximation: exhaustively 
 enumerate over all subsets of size at most three and extend each using
 the 
  greedy heuristic [Sviridenko, 2004].
{We present a novel refined analysis of the greedy heuristic which proves that it suffices to enumerate only over all subsets of size at most two and still retain  a tight 
 $(1-e^{-1})$-approximation.
 This improves the running time from $O(n^5)$ to $O(n^4)$ queries.
}

 }

Many algorithms for maximizing a monotone submodular function subject to a knapsack constraint rely on the natural greedy heuristic.
We present a novel refined analysis of this greedy heuristic which enables us to:
$(1)$ reduce the enumeration in the tight $(1-e^{-1})$-approximation of [Sviridenko 04] from subsets of size three to two;
$(2)$ 
present an improved upper bound of $0.42945$  for the classic algorithm which returns the better between 
a single element and the output of the greedy heuristic.

\comment{
provide upper bounds on classic heuristics for the problem.  

$(1)$ it suffices to enumerate over subsets of size at most two and extend each using the greedy heuristic to obtain a tight $(1-e^{-1})$-approximation for the problem ;

$(2)$ }
\end{abstract}


\section{Introduction}

Submodularity is a fundamental mathematical notion that captures the concept of economy of scale and is prevalent in many areas of science and technology.
Given a ground set $E$, a set function $f:2^E \to \mathbb{R}$ over $E$ is called \emph{submodular} if it has the \emph{diminishing returns} property:
$f(A \cup \{e\}) - f(A) \geq f(B \cup \{e\}) - f(B)$ for every $A \subseteq B \subseteq E$ and $e \in E \setminus B$.\footnote{
	An equivalent definition is: $f(A) + f(B) \geq f(A \cup B) + f(A \cap B)$ for any $A,B \subseteq E$.
}
Submodular functions naturally arise in different areas such as combinatorics, graph theory, probability, game theory, and economics.
Some well known examples include coverage functions, cuts in graphs and hypergraphs, matroid rank functions, entropy, and budget additive functions.

A submodular function $f$ is monotone if $f(S) \leq f(T)$ for every $S \subseteq T \subseteq E$.
In this note we consider the problem of {\em maximizing a monotone submodular function subject to a knapsack constraint} (MSK).
An instance of the problem is a tuple $(E,f,w,W)$ where 
$E$  is a set of $n$ elements, $f:2^{E}\rightarrow \mathbb{R}_{\geq 0}$ is a non-negative, monotone
  and submodular set function given by a {\em value oracle}, $w:E \rightarrow \mathbb{N}_+$ is a  weight function  over the elements, and $W\in \mathbb{N}$ is the knapsack capacity.\footnote{We use $\mathbb{N}$ to denote the set of non-negative integers, and $\mathbb{N}_+=\mathbb{N}\setminus\{0\}$.} A subset $S\subseteq E$ is {\em feasible} if $\sum_{e\in S} w(e) \leq W$, i.e., the total weight of elements in $S$ does not exceed the capacity $W$; the value of $S\subseteq E$ is $f(S)$. The objective is to find a feasible subset $S\subseteq E$ of maximal value.

MSK arises in many applications. Some examples include sensor placement \cite{LKGFVG07}, document summarization \cite{LB10}, and network optimization \cite{SKYK11}.
The problem is a generalization of monotone submodular maximization with a cardinality constraint (i.e., $w(e)=1$ for all $e\in E$), for which a simple greedy algorithm yields a $(1-e^{-1})$-approximation \cite{NWF78}. 
This is the best ratio which can be obtained in 
polynomial time in the oracle model~\cite{NW78}. The approximation ratio of $(1-e^{-1})$ is also optimal in the special case of 
coverage functions  
under $P\neq NP$ \cite{Fe98}.

Many algorithms for MSK rely on a natural greedy heuristic. Greedy
maintains a feasible subset $A\subseteq E$. In each step it adds to $A$ an element $e\in \{e'\in E~|~w(A\cup\{e'\})\leq W\}$ which maximizes $\frac{f\left(A\cup\{e\}\right)-f (A)}{w(e)}$.\footnote{For a set $A\subseteq E$ we use $w(A)=\sum_{e\in A} w(e)$.} While the greedy heuristic does not guarantee any constant approximation ratio, it is commonly utilized as a procedure within approximation algorithms.

The first $(1-e^{-1})$-approximation for MSK was given by Sviredenko \cite{Sv04} as an adaptation of an algorithm of Khuller, Moss and Naor \cite{KMS99} proposed for the special case of coverage functions. 
The algorithm of Sviridenko exhaustively enumerates  (iterates) over all subsets $G\subseteq E$ of at most $3$ elements and extends each set $G$ using the greedy heuristic.
 The algorithm uses $O(n^5)$ oracle calls and arithmetic operations.

Several  works were dedicated to the development of  simple, fast and greedy-based algorithms for MSK, with approximation ratios strictly smaller than  $(1-e^{-1})$.  
Special attention was given to the algorithm which returns the better solution between the single element of highest value and the result of the greedy heuristic, to which we refer as $\greedyplus$ (see the pseudocode in Section \ref{sec:plus}).

The algorithm was first suggested in \cite{KMS99} for coverage functions, and adapted to monotone submodular function in \cite{LB10}. Both
works stated an approximation guarantee of $(1-e^{-0.5})$, though the proofs in both works were flawed. 
A correct proof  for a $(1-e^{-0.5})$-approximation was given by Tang et. al.  \cite{tang2021revisiting},
improving  upon an earlier approximation guarantee of  $\frac{e-1}{2e-1}\approx 0.387$  by Cohen and Katzir \cite{cohen2008generalized}.  Recently, Feldman, Nutov and Shoham~\cite{FNS20} showed that the approximation ratio of the algorithm is within $[0.427,0.462]$.

A recent work by  Yaroslavtsev, Zhou and  Avdiukhin \cite{yaroslavtsev2020bring} gives an $O(n^2)$ greedy based algorithm with an
approximation guarantee of $\frac{1}{2}$. A variant of the algorithm of Yaroslavtsev el. al. was used by Feldman et. al. \cite{FNS20} to derive a greedy based $0.9767\cdot(1-e^{-1})$-approximation in $O(n^3)$ oracle calls.

Taking a more theoretical point of view, in \cite{EN19} Ene and Nguyen presented a $(1-e^{-1}-\eps)$-approximation for MSK in time  $O(n \cdot \log^2{n})$ for any fixed $\eps>0$, improving upon an earlier  $O(n^2\cdot \text{polylog} (n))$ algorithm with the same approximation ratio due to Badanidiyuru and Vondr{\'a}k \cite{BV14}. We note, however, that the dependence of the running times of  these algorithms on $\eps$ renders them purely theoretical.

Our main technical contribution is a tighter  analysis of the greedy heuristic, presented in Section~\ref{sec:greedy_alg}. 
We show two applications of the analysis.
In the first application we consider a variant of the algorithm in \cite{Sv04} which only enumerates of subsets of size at most two (as opposed to three in \cite{Sv04}),  and show it retains the tight approximation ratio of $(1-e^{-1})$ for MSK.
\begin{thm}
\label{thm:main}
There is a $(1-e^{-1})$-approximation for MSK using $O(n^4)$ value oracle calls and arithmetic operations which works as follows: enumerate over all subsets of size at most two and extend each using the greedy heuristic.
\end{thm}

Let us now briefly elaborate on the insight we use in order to improve the analysis of \cite{Sv04}.
Intuitively, the analysis in \cite{Sv04} bounds the value of  the solution generated by the greedy phase assuming a worst case submodular function $f$, and then bounds the value loss  due to a discarded element (the element is discarded by the analysis, not by the algorithm) assuming a worst case submodular function $g$. 
The main insight for our improved result is that $g\neq f$; that is, there is no function which 
attains simultaneously the worst cases assumed in \cite{Sv04} for the outcome of greedy and for the value loss due to the discarded element. This insight is well captured by the refined analysis. 
The proof of Theorem~\ref{thm:main} is given in Section~\ref{sec:enum}. 

In the second application we
utilized the refined analysis to improve the best known upper bound on the approximation ratio of $\greedyplus$.  

 \begin{thm}
	\label{thm:gs_bound} The approximation ratio of $\greedyplus$ is no greater than $\beta= 0.42945$.
\end{thm}
 The proof of the above theorem is given in Section \ref{sec:plus}.  The result was obtained by generating an instance for which the guarantee of the refined analysis  is poor via numerical optimization. Combined with the result of \cite{FNS20}, the theorem limits the approximation ratio of $\greedyplus$ to the narrow interval $[0.427,0.4295]$. 

We note that our  observation, which states that it suffices to enumerate only over all subsets of size at most two, was also independently and in parallel obtained by Feldman, Nutov and Shoham \cite{FNS20}.
However, our proof and the proof  of \cite{FNS20} differ significantly.
Our approach for proving the main 
observation is useful for constructing
counter examples for the approximation 
of the greedy heuristic, as demonstrated in Theorem \ref{thm:gs_bound}.



\section{The Greedy Procedure}

We start with some definitions and notation.
Given a monotone submodular function $f:2^E\rightarrow \mathbb{R}_{\geq 0}$ and $A\subseteq E$, we define the function $f_A:2^E\rightarrow \mathbb{R}_{\geq 0}$ by $f_A(S)= f(A\cup S)-f(A)$ for any $S\subseteq E$. It is well known that $f_A$ is also  monotone, submodular and non-negative (see, e.g., Claim 13 in \cite{FKNRS20}).  We also use $f(e)=f(\{e\})$ for $e\in E$.  
\label{sec:greedy_alg}
\begin{algorithm}[h]
	\SetAlgoLined
	\SetKwInOut{Input}{Input}\SetKwInOut{Output}{output}
	\Input{An MSK instance $(E,f,w,W)$}
	\DontPrintSemicolon
	
	
	Set $E' \leftarrow E$ and $A\leftarrow \emptyset$ \;
	
	\While{ $E'\neq \emptyset$ \label{greedy:loop}}{ 
		
		Find $e\in E'$ such that $\frac{f_A(e)}{w(e)}$ is maximal.
		\label{greedy:considered}
		\;
		
		Set $E' \leftarrow E' \setminus \{e\}$.\;
		
		If $w(A\cup \{e\})\leq W$ set $A\leftarrow A\cup \{e\}$. \label{greedy:selected}\; 		
	}
	Return $A$ 
	\;
	\caption{$\greedy(E,f,w,W)$} 
	\label{alg:greedy}
\end{algorithm}

The greedy procedure is given in Algorithm \ref{alg:greedy}.  
While the procedure is useful for deriving efficient approximation, as a stand-alone algorithm it
does not guarantee any constant approximation ratio. 
We say that the element $e\in E$ found in Step \ref{greedy:considered} 
is {\em considered} in the specific iteration of the loop in Step \ref{greedy:loop}.
Furthermore, if the element was also added to $A$ in Step \ref{greedy:selected} we say it was {\em selected} in this iteration. 

\begin{lemma}
	\label{lem:greedy_feasible}
	For any MSK instance $(E,f,w,W)$, Algorithm~\ref{alg:greedy} returns a feasible solution for $(E,f,w,W)$.
\end{lemma}

For any MSK instance $(E,f,w,W)$, we define a {\em value function} $V$. Let $\{a_1,\ldots, a_{\ell}\}$  be the output of $\greedy(E,f,w,W)$,
in the order by which the elements are added to $A$ in Step~\ref{greedy:selected} of Algorithm~\ref{alg:greedy}. Furthermore, define $A_i=\{a_1, \ldots, a_i\}$ for $i\in [\ell]$, and $A_0=\emptyset$.  We define $V:[0,w(A_{\ell})]\rightarrow \mathbb{R}_{\geq 0}$ by\footnote{For any $k\in \mathbb{N}_+$ we use $[k]$ to denote the set $\{i \in \mathbb{N}~|~1\leq i \leq k  \}=\{1,2,\ldots, k\}$.}
$$
\forall i\in [\ell],~w(A_{i-1})\leq u\leq w(A_i):~~~~ V(u) =  f(A_{i-1}) + (u-w(A_{i-1})) \frac{f_{A_{i-1}}(\{a_i\})}{w(a_i)}.$$   We note that the value of $V(w(A_i))$ is well defined for $i\in [\ell -1]$ since
\[
f(A_{i-1}) + \left( w(A_{i})- w(A_{i-1})\right) \frac{f_{A_{i-1}} (\{a_{i}\}) } { w(a_{i})}=f(A_i) = f(A_i ) + \left( w(A_i)- w(A_i)\right) \frac{f_{A_{i}} (\{a_{i+1}\}) } { w(a_{i+1})}.
\]
That is, the value function $V$ is semi-linear and continuous. By definition we have that $V(0)=f(\emptyset)$ and $V(w(A_{\ell}))=f(A_{\ell})$. Intuitively, $V(u)$ can be viewed as the value attained by Algorithm \ref{alg:greedy} while using capacity of $u$.  We use $V'$ to denote the first derivative of $V$. We note that $V'(u)$ is defined for almost all  $u\in [0,w(A_{\ell})]$. Similar to \cite{Sv04}, our analysis is based on lower bounds over $V'$ (in \cite{Sv04} the analysis used a discretization of $V$, thus omitting the differentiation). 

For every $i\in [\ell ]$ and $u\in(w(A_{i-1}),w(A_i))$ we have $V'(u)= \frac{f_{A_{i-1}} (a_i)}{w(a_i)}$. The next lemma gives a lower bound for $V'$.  

\begin{lemma}
	\label{lem:dens_to_Y}
Let $E'$ be the set from Algorithm~\ref{alg:greedy} at the beginning of the  iteration in which $a_i$ is selected, and let $\emptyset\neq Y\subseteq A_{i-1}\cup E'$. Then, 
$\frac{f_{A_{i-1}} (a_i)}{w(a_i)} \geq \frac{f(Y)-f(A_{i-1})}{w(Y)}$. 
\end{lemma} 
\begin{proof}
	Let $m=|Y|$ and  $Y=\{y_1,\ldots, y_m\}$.
	For any $1\leq j \leq m$, if $y_j\in A_{i-1}$ then  $\frac{f_{A_{i-1}}  \left(\left\{y_j \right\}\right)}{w\left(y_j\right)}=0 \leq  \frac{f_{A_{i-1} } \left(\left\{a_{i} \right\}\right)}{w\left(a_{i}\right) }
	$. Otherwise, by the assumption of the lemma,  $y_{j}\in E'$ and therefore $\frac{f_{A_{i-1} } \left(\left\{y_j\right\}\right)}{w\left(y_j\right)} \leq  \frac{f_{A_{i-1} } \left(\left\{a_{i} \right\}\right)}{w\left(a_{i}\right) }$, since $a_{i}$  was selected in Step \ref{greedy:selected} of Algorithm \ref{alg:greedy} when the value of the variable $A$ was $A_{i-1}$. Thus, $\frac{f_{A_{i-1} } \left(\left\{y_j\right\}\right)}{w\left(y_j\right)} \leq  \frac{f_{A_{i-1} } \left(\left\{a_{i} \right\}\right)}{w\left(a_{i}\right) }
	$ for every $j\in [m]$. By the last inequality, and since $f$ is monotone and submodular, we have the following.
	\begin{equation*}
	\begin{aligned}
	f\left(Y\right)  &\leq f\left(A_{i-1} \cup Y\right) \\
	&= f(A_{i-1})+\sum_{j=1}^{m} f_{A_{i-1} \cup \left\{y_{1},\ldots ,y_{j-1}\right\}} \left(\left\{y_{j}\right\} \right)\\
	&\leq f(A_{i-1})+\sum_{j=1}^{m} f_{A_{i-1}} \left(\left\{y_{j} \right\}\right)\\
	&= f(A_{i-1}) + \sum_{j=1}^{m} w\left(y_{j}\right) \frac{f_{A_{i-1} } \left(\left\{y_j \right\}\right)}{w\left(y_j\right)} \\
	&\leq f(A_{i-1} ) + \sum_{j=1}^{m} w\left(y_j\right) \frac{f_{A_{i-1} } \left(\left\{ a_{i} \right\}\right)}{w\left( a_{i}\right)} \\
	&=  f(A_{i-1})+ w(Y)  \frac{f_{A_{i-1}} \left(\left\{ a_{i} \right\}\right)}{w\left( a_{i}\right)}.
	\end{aligned}
	\end{equation*} 
By rearranging the terms, we have $\frac{f_{A_{i-1}} (a_i)}{w(a_i)} \geq \frac{f(Y)-f(A_{i-1})}{w(Y)}$, as desired.
\end{proof}

To lower bound $V$, we use
Lemma \ref{lem:dens_to_Y} with several different sets as $Y$. Let $X$  be a solution for an MSK instance $I=(E,f,w,W)$ and  $X_1,\ldots, X_k$  be a partition of $X$  such that $\frac{f_{X_1\cup \ldots \cup X_{i-1}} (X_i)}{w(X_i)}\geq \frac{f_{X_1\cup \ldots \cup X_{i}} (X_{i+1})}{w(X_{i+1})}$ for every $i\in [k-1]$.
 It makes sense to first use Lemma \ref{lem:dens_to_Y} with $Y=X_1$  and utilize the differential inequality to lower bound $V$ on an interval $[0,D_1]$. The point $D_1$ is set such that, on the interval $[D_1,D_2]$, using Lemma \ref{lem:dens_to_Y} with $Y=X_1\cup X_2$ yields a better lower bound on $V'$ in comparison to $Y=X_1$.  Subsequently, the resulting differential inequality is used to bound $V$ on $[D_1, D_2]$. When repeated $k$ times, the process results in the {\em bounding function} $h$, formally given in Definition \ref{def:bounding}. Lemma~\ref{lem:bound} shows that indeed $h$ lower bounds $V$.

\begin{defn}
	\label{def:bounding}
Let $I=(E,f,w,W)$ be an MSK instance, and 
consider $X\subseteq E$ such that $w(X)\leq W$, and $X_1,\ldots, X_k$ is a partition of $X$ ($X_i\neq \emptyset$ for all $i\in [k]$). Denote $S_j= \bigcup_{i=1}^{j} X_i$ ($S_0=\emptyset$) and $r_j =\frac{f_{S_{j-1}}(X_j)}{w(X_j)}$ for $j\in[k]$. 
 Also,  assume that $r_{1}\geq r_{2}\geq \ldots \geq r_{k}$ and define $D_0=0$, $D_j =\sum_{i=1}^{j} w(X_{i}) \ln \frac{r_{i}}{r_{j+1}}$ for $0\leq j \leq k-1$, and $D_k=\infty$.\footnote{We define $\ln \frac{0}{0}=\ln\frac{a}{0}=\infty$.}  The {\em bounding function} of $X_1,\ldots, X_k$ and $I$ is $h:\mathbb{R}_{\geq 0}\rightarrow \mathbb{R}_{\geq 0}$ defined by 
 \begin{equation}
 \label{eq:bounding}
 \forall j\in[k],~D_{j-1}\leq u < D_j:~~~~
 h(u) = f(S_j) -r_j\cdot w(S_j)\cdot \exp\left(- \frac{u-D_{j-1}}{w(S_j)}\right)
 \end{equation}	
 \end{defn}
It can be easily verified that $D_0\leq D_1\leq \ldots \leq D_k$, and  $D_{j'}=D_{j}$ for $j'<j$ if and only if $r_{j'+1}=r_{j+1}$. 
By definition, the bounding function is differentiable  almost everywhere. Furthermore, for every $j\in[k]$ such that $D_j<D_{j+1}$ and $D_j\neq 0$, let $j'$ be the minimal $j'\in[j]$ such that $D_{j'}= D_j$. Then,  
\begin{equation*}
\begin{aligned}
\lim_{u\nearrow D_j} &h(u) =
f(S_{j'})-r_{j'}\cdot w(S_{j'}) \cdot \exp\left(- \frac{D_{j'}-D_{j'-1}}{w(S_{j'})}\right)\\
&=
f(S_{j'}) -r_{j'}\cdot w(S_{j'})  \cdot \exp\left(- \frac{\sum_{i=1}^{j'} w(X_{i}) \ln \frac{r_{i}}{r_{j'+1}} - \sum_{i=1}^{j'-1} w(X_{i}) \ln \frac{r_{i}}{r_{j'}} }{w(S_{j'})}\right)\\
&=
f(S_{j'})  -r_{j'}\cdot w(S_{j'})  \cdot \exp\left( 
- \frac{\sum_{i=1}^{j'} w(X_{i})}{w(S_{j'})}\cdot \ln \frac{r_{j'}}{r_{j'+1}}
\right)
\\
&=f(S_{j'})  -r_{j'+1} \cdot w(S_{j'})  \\
&= f(S_{j'})  +\sum_{i=j'+1}^{j+1} f_{S_{i-1}} (X_i) - \sum_{i=j'+1}^{j+1} r_{i} \cdot w(X_i) -r_{j+1} \cdot w(S_{j'}) \\
&= f(S_{j+1}) - r_{j+1} \cdot w(S_{j+1})=
 h(D_j).
\end{aligned}
\end{equation*}
The first equality follows from \eqref{eq:bounding}. The second and forth equalities follow from the definitions of $D_{j'}$ and $S_{j'}$. The sixth equality holds since for every $j'< i \leq j+1$ we have that
 $\frac{f_{S_{i-1}}\left(X_{i} \right) }{w\left(X_{i}\right)} = r_{i}= r_{j+1}$. 
Thus, the bounding function $h$ is continuous.

Let $S^*$ be an optimal solution for an MSK instance, $I=(E,f,w,W)$.
 Then, the bounding function of $I$ 
 and $S^*$ (i.e., $k=1$) is 
$h(u)= f(S^*)\left(1-\exp\left(-\frac{u}{W} \right)\right)$. It follows from \cite{Sv04} that $V(u)\geq\left(1-\exp\left(-\frac{u}{W} \right)\right)f(S^*)= h(u)f(S^*)$ for $u\in [0, W-\max_{e\in S^*} w(e)] \cap \mathbb{N}$, where the restriction to integer values can be easily relaxed. Thus, the following lemma can be viewed as a  generalization of the analysis of~\cite{Sv04}.

\begin{lemma}
	\label{lem:bound}
	
Let $I=(E,f,w,W)$ be an MSK instance, $V$ its value function, and $A$ the output of Algorithm \ref{alg:greedy} for the instance $I$. Consider a subset of elements $X\subseteq E$ where $w(X)\leq W$, and 
a partition $X_1,\ldots, X_k$ of $X$, such that $\frac{f_{X_1\cup \ldots \cup X_{i-1}} (X_i)}{w(X_i)}\geq \frac{f_{X_1\cup \ldots \cup X_{i}} (X_{i+1})}{w(X_{i+1})}$ for any $i\in [k-1]$. Let $h$ be the bounding function of $I$ and $X_1,\ldots, X_k$, and $W_{\max}=\min\left\{ W-\max_{e\in X} w(e), ~ w(A)\right\}$. Then, for any $u\in [0, W_{\max}]$, it holds that $V(u)\geq h(u)$. 
\end{lemma}
 The proof of Lemma \ref{lem:bound} uses a differential comparison argument. We say a  function $\varphi:Z\rightarrow \mathbb{R}$, $Z\subseteq \mathbb{R}^2$, is {\em positively linear in the second dimension} if there is $K>0$ such that for any $u, t_1, t_2$ where $(u,t_1), (u,t_2)\in Z$  it holds that $\varphi(u,t_1)-\varphi(u,t_2)= K\cdot \left(t_1-t_2\right)$. The following is a simple variant of standard differential comparison theorems (see, e.g., \cite{Mc86}). 
\begin{lemma}
	\label{lem:comparison}
Let $[a,b]=\bigcup_{r=1}^{s} [c_{r},c_{r+1}]$ be an interval such that $c_1\leq c_2\leq \ldots \leq c_{s+1}$ and $\vartheta_1, \vartheta_2:[a,b]\rightarrow \mathbb{R}$ be two continuous functions  such that $\vartheta_1(a)\geq \vartheta_2(a)$ and the deriviatives $\vartheta'_1,\vartheta'_2$ are defined and continuous on $(c_r,c_{r+1})$ for every $r\in [s]$. Also, for any $r\in [s]$ let $\varphi_r:(c_r,c_{r+1})\times \mathbb{R} \rightarrow \mathbb{R}$ be positively linear in the second dimension.
  If $\vartheta_1'(u)\geq \varphi_r(u, \vartheta_1(u))$ and $\vartheta_2'(u)\leq \varphi_r(u, \vartheta_2(u))$ for every $r\in [s]$ and $u\in (c_r,c_{r+1})$, then $\vartheta_1(u) \geq \vartheta_2(u)$ for every $u\in [a,b]$. 
\end{lemma}
The lemma follows from standard arguments in the theory of differential equations. A formal proof is given in Appendix \ref{sec:app}
\begin{proof}[Proof of Lemma \ref{lem:bound}]
Let $(D_j)_{j=0}^k$,  and  $\left(S_j\right)_{j=0}^{k}$ be as in Definition \ref{def:bounding}.
Define 
\begin{equation}
\label{eq:phi_def}
\varphi(u,v)= \frac{f(S_j) - v}{w(S_j)}
\end{equation}
for $j\in[k]$, $D_{j-1}\leq u<D_j$ and $v\in\mathbb{R}$. 
Let $A=\{a_1,\ldots, a_{\ell}\}$, where $a_1,\ldots, a_{\ell}$ is the order by which the elements were added to $A$ in Step \ref{greedy:selected} of Algorithm \ref{alg:greedy}. As before, we use $A_i=\{a_1,\ldots, a_i\}$ for $i\in [\ell]$ and $A_0=\emptyset$. 
Let $C= (0,W_{\max})\setminus \{D_1,\ldots, D_{k-1}\}\setminus \{a_1,\ldots, a_{\ell}\}$. Then, for any $u \in C$, 
there is $j\in[k]$ such that $D_{j-1}<u<D_j$. Hence,
\begin{equation}
\label{eq:h_diff}
\varphi(u,h(u))= \frac{f(S_j) - h(u)}{w(S_j)} = 
r_{j}\cdot \exp\left( -\frac{u-D_{j-1}}{w(S_j)}\right)
=h'(u),
\end{equation}
where $h'$ is the first derivative of $h$.  
As $u\in C$, there is also $i\in [\ell]$ such that $w(A_{i-1})<u<w(A_i)$. Hence,
\begin{equation}
\label{eq:V_diff}
\begin{aligned}
V'(u)
= \frac{f_{A_{i-1}}(a_i)}{w(a_i)}
\geq 
\frac{f(S_j)- f(A_{i-1})}{w(S_j)}
=  \frac{f(S_j)- V(w(A_{i-1}))}{w(S_j)}
\geq \frac{f(S_j)- V(u)}{w(S_j)}
=\varphi(u,V(u))
\end{aligned}
\end{equation}
For the first inequality, we note that $X \subseteq A_{i-1}\cup E'$,
where $E'$ is the set at the beginning of the iteration in which $a_i$ was selected. Indeed, otherwise we have that $X$ contains an element $e \in E$ that was considered by the algorithm at some iteration $1 \leq \ell < i$, but not selected  since $w(A_\ell \cup \{ e \}) > W$. This would imply that $u > w(A_{i-1}) \geq w(A_\ell) > W_{max}$. Thus, as $S_j \subseteq X$ we have the conditions of Lemma~\ref{lem:dens_to_Y}.
 The second inequality holds since $V$ is increasing. By \eqref{eq:h_diff} and \eqref{eq:V_diff} we have
\begin{equation}
\label{eq:forall_u}
\forall u\in C :~~~ h'(u)=\varphi(u, h(u)) \textnormal{~~~and~~~} V'(u)\geq \varphi(u,V(u)).
\end{equation}

We can write $C=\bigcup_{r=1}^{s} (c_r, c_{r+1})$ where $0=c_1\leq c_2 \leq \ldots \leq c_{s+1}=W_{\max}$.
For any $r\in [s]$ let $\varphi_r:(c_r,c_{r+1})\rightarrow\mathbb{R}$ be the restriction of $\varphi$ to $(c_r, c_{r+1})$ ($\varphi_r(u)=\varphi(u)$ for any $u\in (c_r, c_{r+1})$). It can be easily verified that $\varphi_r$ is continuous and positively linear in the second dimension. Furthermore, it holds that  $V'$ and $h'$ are continuous on $(c_r,c_{r+1})$ for any $r\in [s]$. Thus, by \eqref{eq:forall_u} and Lemma \ref{lem:comparison} it holds that $V(u)\geq h(u)$ for any $u\in [0,W_{\max}]$.
\end{proof}

\section{ A Faster Greedy Based Algorithm}
\label{sec:enum}
\begin{algorithm}[h]
	\SetAlgoLined
	\SetKwInOut{Input}{Input}
	\Input{An MSK instance $(E,f,w,W)$, and enumeration size $\kappa \in \mathbb{N}$.}
	\DontPrintSemicolon
	Set $S^*\leftarrow \emptyset$ \;
	\For{every  $G\subseteq E$, $|G|\leq \kappa$ \label{enum:loop}}{ 
		
		$A\leftarrow \greedy(E,f_G, w, W-w(G))$
		\label{enum:greedy}
		\;
		If $f(A\cup G)\geq f(S^*)$ then $S^*\leftarrow A \cup G$ \label{enum:update}\;
		
	}
	Return $S^*$ 
	\;
	\caption{$\enum_{\kappa}(E,f,w,W)$} 
	\label{alg:enum}
\end{algorithm}

To prove Theorem \ref{thm:main} we use $\enum_2$, i.e., we take Algorithm~\ref{alg:enum} with $\kappa = 2$.
We note that $\enum_3$ is the $(1-e^{-1})$-approximation algorithm 
of~\cite{Sv04}.  

\begin{lemma}
\label{lem:enum2}
$\enum_{2}$ is a $(1-e^{-1})$-approximation  for MSK.
\end{lemma}
\begin{proof}
It can be easily verified that the algorithm always returns a feasible solution for the input instance. 
Let $(E,f,w,W)$ be an MSK instance and $Y\subseteq E$ an optimal solution for the instance. Let $Y=\{y_1,\ldots, y_{|Y|}\}$, and assume the elements are ordered by their marginal values. That is, $f_{\{y_1,\ldots, y_{i-1}\}} \left( \{y_i\}\right)= \max_{1\leq j \leq |Y| }f_{\{y_1,\ldots, y_{i-1}\}} \left( \{y_j\}\right)$ for every $1\leq i\leq |Y|$.  

If $|Y|\leq 3$ then there is an iteration of Step \ref{enum:loop} in which $G=\{y_i~|~i\in \{1,2\},~i\leq |Y|~\}\}$. In this iteration it holds that $f(G\cup A)\geq f(G)\geq \frac{2}{3}f(Y)$ (since $f$ is monotone); thus, following this iteration we have  $f(S^*)\geq \frac{2}{3}f(Y)\geq (1-e^{-1})f(Y)$. Therefore, in this special case the algorithm  returns an approximate solution as required. Hence, we may assume that $|Y|>	 3$.

Our analysis focuses on the iteration of the loop of Step \ref{enum:loop} in which $G=\{y_1, y_2\}$. Let $A$ be the output of $\greedy(E,f_G, w, W-w(G))$ in Step \ref{enum:greedy} in this iteration. If $Y\setminus G \subseteq A$ then $f(A\cup G) \geq f(Y)$, thus following this iteration it holds that $f(S^*)\geq f(Y)$, and the algorithm returns an optimal solution. Therefore, we may assume that
$Y\setminus G\not\subseteq A$.

\comment{
\ariel{I am not sure it is a good idea to add this explanation in the middle of the proof, but I do not have a better place}
We note that up to this point the analysis is similar to the one  of $\enum_3$ in \cite{Sv04}. In \cite{Sv04} it is shown that there is 
$e\in Y\setminus G$ such that 
$f_G(A\cup \{e\}) \geq f_G(Y)\left(1-e^{-1}\right)$. Thus
$f(G)+f_A(G)\geq f(G)+f_G(Y)\left(1-e^{-1}\right) -f_G(\{e\})$. The bound $f_G(\{e\})\leq \frac{1}{3}f(G)$ is then used to derive the approximation ratio. The intuition to our proof is that if $f_G(\{e\})\approx \frac{1}{3}f(G)$ then $f$ is behaves similarly to a modular (linear) function, thus, $f_G(A)$ should be significantly greater than $f_G(Y)\left(1-e^{-1}\right)$

Back to our proof,} 
Let $e^*\in Y\setminus G$ such that $w(e^*)=\max_{e\in Y\setminus G} w(e)$ and denote $R=Y\setminus G \setminus \{e^*\}$. Define two sets $X_1, X_2$ such that  $\{X_1, X_2\}=\left\{\{e^*\}, R\right\}$ and 
$\frac{f(X_1)}{w(X_1)}\geq \frac{f(X_2)}{w(X_2)}$. As $f$ is submodular it follows that $\frac{f_{X_1}(X_2) }{w(X_2)}\leq \frac{f(X_2)}{w(X_2)}\leq \frac{f(X_1)}{w(X_1)}$.
Let $h$ be the bounding function of $(E, f_G, w, W-w(G))$ and  $X_1, X_2$. Also, let $r_1, r_2$,  and $D_1$ be the values from Definition \ref{def:bounding}.

\comment{
\ariel{fix this}
As $y_1,\ldots, y_{|Y|}$ are sorted by their marginal values, it holds that $f(y_1), f_{\{y_1\}}(y_2)\geq f_G(e^*)$. Hence, 
\begin{equation}
\label{eq:case1_estar}
3\cdot f_G(X_1)= 3\cdot f_{G}(\{e^*\}) \leq f(y_1)+f_{\{y_1\}}(y_2)+f_G(\{e^*\})= f(Y)-f_{G\cup X_1} (X_2).
\end{equation}
And similarly, 
\begin{equation}
\label{eq:case3_ineq}
\frac{3}{2}f(G) \geq    f(G) +f_G(X_1) = f(Y)-f_{G\cup X_1}(X_2) .
\end{equation} 
}

By Step \ref{greedy:selected} of Algorithm~\ref{alg:greedy}, as $Y\setminus G\not\subseteq A$, it follows that $w(A)\geq W-w(G)-w(e^*)$. Thus, by Lemma \ref{lem:bound}, it holds that $f_G(A)\geq V(W-w(G)-w(e^*))
\geq h(W-w(G)-w(e^*))$. We consider the following cases.

\paragraph{Case 1: $W-w(G)-w(e^*)\geq D_1$.}  In this case it holds that
\begin{equation}
\label{eq:case1_initial}
\begin{aligned}
f(G)&+f_G(A)\geq f(G)+h(W-w(G)-w(e^*)) \\
& =f(G)+f_G(X_1\cup X_2) -w(X_1\cup X_2)\cdot r_{2} \cdot \exp\left( -\frac{W-w(G)-w(e^*) -D_1 }{w(X_1 \cup X_2)} \right)\\
&= f(Y)   - w(X_1\cup X_2)\cdot \exp\left(
-\frac{W-w(G)-w(e^*) -w(X_{1}) \cdot \ln \frac{ r_{1}}{r_{2}} }{w(X_1\cup X_2)} +\ln r_{2}\right)\\
&\geq 
f(Y)   -w(X_1\cup X_2)\cdot \exp\left(
-1  +\frac{w(e^*)+w(X_{1}) \cdot \ln  r_{1} +
  w(X_{2})\cdot \ln r_{2}
}{w(X_1\cup X_2)}\right).\\
\end{aligned}
\end{equation}
The first and second equalities follow from the definitions of $h$ and $D_1$ (Definition \ref{def:bounding}). The last inequality follows from $w(X_1\cup X_2) + w(G)\leq W$. Define  two sets $H_{e^*}, H_R$ as follows. If $X_1= \{e^*\}$ then $H_{e^*}=\emptyset$ and $H_R=\{e^*\}$. If $X_1= R$ then $H_{e^*}= R$ and $H_R=\emptyset$. It follows that
\begin{equation}
\label{eq:H_first}
\begin{aligned}
w(X_1) \cdot \ln  r_{1} &+
w(X_2)\cdot \ln r_2 
=w(X_1) \cdot \ln  \frac{f_G(X_1)}{w(X_1)} +
w(X_2)\cdot \ln \frac{f_{G\cup X_1} (X_2)}{w(X_2)} \\ &=
w(e^*) \cdot \ln \frac{ f_{G\cup H_{e^*}} (e^*) }{w(e^*)}  + w(R)\cdot  \ln \frac{f_{G\cup H_R} (R)}{w(R)}.
\end{aligned}
\end{equation}
As the elements  $y_1,\ldots, y_m$ are ordered according to their  marginal values, we have that $f(y_1)\geq f_{y_1}(y_2)\geq f_{G}(e^*)\geq f_{G\cup H_{e^*}}(e^*)$. Therefore, $f(G)\geq 2\cdot f_{G\cup H_{e^*}}(e^*)$ and we have that
\begin{equation}
\label{eq:H_second}
f(Y)-f_{G\cup H_R}(R)=f(G)+f_{G\cup  H_{e^*}}(e^*)\geq 3 \cdot f_{G\cup  H_{e^*}}(e^*).
\end{equation}
By combining \eqref{eq:H_first}
 and \eqref{eq:H_second} we obtain the following.
 \begin{equation}
 \label{eq:H_third}
 \begin{aligned}
 w(X_1) \cdot \ln  r_{1} &+
 w(X_2)\cdot \ln r_2 \leq
w(e^*) \cdot \ln \frac{ f(Y)-f_{G\cup H_R} (R) }{3\cdot w(e^*)}  + w(R)\cdot  \ln \frac{f_{G\cup H_R} (R)}{w(R)}\\
&=-w(e^*)\cdot \ln 3 + w(e^*) \cdot \ln \frac{ f(Y)-f_{G\cup H_R} (R) }{ w(e^*)}  + w(R)\cdot  \ln \frac{f_{G\cup H_R} (R)}{w(R)} \\
&\leq -w(e^*)+ w(R\cup \{e^*\})\ln \frac{ f(Y)}{w(R\cup \{e^*\})}
= -w(e^*) +w(X_1 \cup X_2 ) \ln \frac{ f(Y)}{w(X_1\cup X_2)}
\end{aligned}
 \end{equation}
 The second inequality follows from the log-sum inequality (see, e.g, Theorem 2.7.1 in \cite{Co06}) and $\ln 3 >1$. 

Plugging \eqref{eq:H_third} into \eqref{eq:case1_initial} we obtain the next inequality.
\begin{equation*}
f(G)+f_G(A) \geq f(Y) - w(X_1\cup X_2)\cdot \exp\left(-1 +\ln\frac{f(Y)}{w(X_1\cup X_2)}\right)  
=f(Y) \cdot \left(1-e^{-1}\right).
\end{equation*}

\paragraph{Case 2: $W-w(G)-w(e^*)< D_1$ and $X_1=\{e^*\}$.}    We can use the assumption in this case to lower bound $\frac{f_G(X_1)}{f_{G\cup X_1} (X_2)}$ as follows.
$$
W-w(G)-w(X_1) < D_1=w(X_1)\cdot \ln \frac{r_1}{r_2} 
=w(X_1)\cdot \left(  
 \ln \frac{f_G(X_1)}{f_{G\cup X_1} (X_2)} +\ln \frac{w(X_2)}{w(X_1)}
 \right)
$$
By rearranging the terms we have
$$
\ln \frac{f_G(X_1)}{f_{G\cup X_1} (X_2)} >
\frac{ W-w(G)-w(X_1)} {w(X_1)}
 -\ln \frac{w(X_2)}{w(X_1)}.
$$
Thus,
\begin{equation}
\label{eq:case_2_f1_lb}
f_G(X_1) > {f_{G\cup X_1} (X_2)}  \cdot \frac{w(X_1)}{w(X_2)} 
\cdot \exp\left(
\frac{ W-w(G)-w(X_1)} {w(X_1)}
\right)\geq 
{f_{G\cup X_1} (X_2)}  \cdot \delta^{-1}
\cdot \exp\left(
\delta
\right),
\end{equation}
where $\delta =\frac{W-w(G)-w(X_1)}{w(X_1)}$ and the last inequality follows from $w(X_1)+w(X_2)+w(G)\leq W$. 

We use \eqref{eq:case_2_f1_lb} to lower bound $f(G\cup A)$.
\begin{equation*}
\begin{aligned}
f(G)&+f_G(A)\geq f(G)+ h (W-w(G)-w(X_1))\\
&
= f(G)+f_{G}(X_1)-f_G(X_1) \cdot\exp \left(-\delta \right) \\
&= \frac{2}{3}\left( f(G)+f_{G}(X_1)\right) + \frac{1}{3} \left( f(G)+f_{G}(X_1)\right)   -f_G(X_1) \cdot\exp \left(-\delta\right)\\
&\geq 
\frac{2}{3}\left( f(G)+f_{G}(X_1)\right) + f_{G}(X_1)  -f_G(X_1) \cdot\exp \left(-\delta \right) \\
&\geq 
\frac{2}{3}\left( f(G)+f_{G}(X_1)\right) + f_{G\cup X_1}(X_2)\left(1-\exp\left(-\delta\right)\right)\cdot \delta^{-1}\cdot \exp\left(\delta \right)\\
&\geq \frac{2}{3}\left( f(G)+f_G(X_1)+f_{G\cup X_1}(X_2)\right) \geq \left(1-e^{-1}\right)f(Y)
\end{aligned}
\end{equation*}
The second inequality follows from $f(G)=f(y_1)+f_{\{y_1\}}(y_2) \geq 2\cdot f_{G}(e^*)= 2 f_G(X_1)$ due to the ordering of elements in $Y$. The third inequality follows from \eqref{eq:case_2_f1_lb}. The forth inequality follows from 
$$\left(1-\exp(-\delta)\right) \cdot \delta^{-1} \cdot \exp(\delta)
= \left(\exp(\delta)-1\right)\cdot {\delta^{-1}}\geq 1\geq \frac{2}{3},$$
as $\exp(\delta)\geq 1+\delta$.

\paragraph{Case 3: $W-w(G)-w(e^*)< D_1$ and $X_1=R$.}
In this case, we have
\begin{equation*}
\begin{aligned}
f(G)&+f_G(A)\geq f(G)+ h (W-w(G)-w(e^*))\\
&= f(G)+f_{G\ }(R) - f_{G}(R) \cdot \exp\left( -\frac{W-w(G)-w(e^*)}{w(R)} \right) \\
&
\geq \frac{2}{3}\left(f(Y) - f_{G}(R)\right)+f_{G}(R) - f_{G}(R) \cdot \exp\left(-1\right) \\
&\geq \left( 1-e^{-1}\right)f(Y)
\end{aligned}
\end{equation*}
The second inequality follows from  $w(X_1)+w(X_2)+w(G)\leq W$, and $\frac{3}{2}f(G)\geq f(G)+f_{G\cup R}(\{e\})= f(Y)-f_G(R)$, as $G=\{y_1, y_2\}$.  

Thus, in all cases $f(A\cup G)=f(G)+f_G(A)\geq (1-e^{-1})f(Y)$. Hence,
in the iteration where $G=\{ y_1, y_2\}$ we have
that $f(S^*)\geq \left(1-e^{-1}\right) f(Y)$.
\end{proof}
 Theorem \ref{thm:main} follows  from Lemma \ref{lem:enum2} and the observation that $\enum_2$ uses $O(n^4)$ oracle calls and arithmetic operations. It is natural to ask whether $\enum_1$ also yields a $(1-e^{-1})$-approximation. 
 Here, the answer is clearly negative. For any $N>0$, consider the MSK instance
 $I=(E,f,w,W)$, with $E=\{1,2,3\}$, $w(1)=w(2)=N$, $w(3)=1$, $W= 2N$ and $f(S)= | S\cap \{1,2\}| \cdot N + 2\cdot |S\cap\{3\}|$. While the optimal solution for the instance is $\{1,2\}$ with $f(\{1,2\})=2N$, $\enum_1(E,f,W,w)$ returns either $\{1,3\}$ or $\{2,3\}$ where $f({1,3})=f(\{2,3\})=N+1$. Already for $N=4$ the solution returned is not an $(1-e^{-1})$-approximation. We note that the function $f$ in this example is modular (linear).
 
 \section{Upper Bound on $\greedyplus$}
 \label{sec:plus}
 
 
In this section we consider the common heuristic $\greedyplus$ for MSK (see Algorithm \ref{alg:greedy_plus}).
We construct an MSK instance for which the approximation ratio of $\greedyplus$ is strictly smaller than $0.42945$, thus tightening the upper bound and almost matching the lower bound of~\cite{FNS20}.
The input instance was generated by numerical optimization 
using an alternative formulation of Lemma \ref{lem:bound}. The numerical optimization process was based on a grid search combined with quasi-convex optimization. As the grid search does not guarantee an optimal solution, 
applying the same approach with an improved numerical optimization may  
lead to a tighter bound.

 \begin{algorithm}[h]
 	\SetAlgoLined
 	\SetKwInOut{Input}{Input}
 	\Input{An MSK instance $(E,f,w,W)$}
 	\DontPrintSemicolon 	
 	$A\leftarrow \greedy(E,f, w, W)$ 
 	\label{plus:greedy}
 	\;
 	Let $e = \argmax_{e\in E:~w(e)\leq W} f(\{e\})$\;
 	If $f(A)\geq f(\{e\})$ return $A$, otherwise return $\{e\}$\;
 	\caption{$\greedyplus$} 
 	\label{alg:greedy_plus}
 \end{algorithm}

\begin{proof}[Proof of Theorem \ref{thm:gs_bound}]

Let $\mu$ be the Lebesgue measure on $\mathbb{R}$. That is, 
given $\II$, a union of intervals on the real-line,
$\mu(\II)$ is the ``length'' of $\II$. In particular, for $\II = [a,b]$,
$\mu(\II) = b-a$. 
We define an MSK  instance  $(E,f, w,W)$ in which $E\subseteq 2^{\mathbb{R}}$, all the sets in $E$ are measurable, and  for any $S \subseteq E$, $f(S)= \mu(\bigcup_{I\in S} I)$. It is easy to verify  that $f$ is submodular and monotone.
Define $f_X= f_Y = 0.42943$, $f_Z= 1-f_X-f_Y=0.14114$, 
$X=[0, f_X]$, $Y=[1,1+f_Y]$ and $Z=[2,2+f_Z]$. Also, define $w(X)=w(Y)= 0.4584$, $w(Z)=1-w(X)-w(Y)= 0.0832$ and $W=1$.
 
The optimal solution for the instance  is $\{X,Y,Z\}$ with $f(\{X,Y,Z\})=1$  and $w(\{X,Y,Z\})=1=W$. In the following we add elements to the instance
such that the greedy algorithm (invoked in Step \ref{plus:greedy} of Algorithm \ref{alg:greedy_plus}) selects the other elements rather than 
$\{X,Y\}$. The greedy algorithm  first selects elements which are subsets of $Z$, then elements which are subsets of $X\cup Y\cup Z$, and finally it selects $Z$.

Define $\delta = 10^{-20}$ and $\rho= 0.62233$. 
There is 
 $0<\eps<\beta$ such that  
$k_1 = {\eps^{-1}\left(\mu(Z)- \frac{\mu(X)}{w(X)} w(Z) \right)}$  and 
 $k_2 = {\eps^{-1} \left( \frac{\mu(X)}{w(X) } - \rho\right) }+k_1$ are integral numbers,   
\begin{equation}
	\label{eq:eps_bound1}
	w(Z)\cdot \ln \left(\mu(Z) \cdot \frac{1 }{w(Z)\cdot \frac{\mu(X)}{w(X)}+\eps }\right) > 
	w(Z)\cdot \ln \left(\frac{\mu(Z) \cdot w(X) }{w(Z)\cdot \mu(X)}\right) - \delta
	\end{equation} and \begin{equation}
	\label{eq:eps_bound2}
	\ln \left( \frac{ \mu(X)}{\rho \cdot w(X) +\eps \cdot w(X) }\right)>\ln \left( \frac{ \mu(X)}{ \rho \cdot w(X) }\right) - \delta. 
	\end{equation}
  Let $(I_j)_{j=1}^{k_1}$ be a sequence of disjoint sets such that 
$\mu(I_j)=\eps$ and $I_j\subseteq Z$ for every $1\leq j \leq k_1$. 
Also, define $w(I_j )=\eps \cdot \frac{w(Z)}{\mu(Z)-\eps (j-1)} = \eps \cdot \frac{w(Z)}{ \mu(Z\setminus I_1 \setminus \ldots \setminus I_{j-1})}$ for $1\leq j\leq k_1$.
Such sets exist since
 $k_1 \cdot \eps = \mu(Z) -\frac{\mu(X)}{w(X) } \cdot w(Z) \approx 0.06319$ .
 
 Define $\mathcal{I}_1= \bigcup_{j=1}^{k_1} I_j$. It holds that,  $$w(\II_1)=\sum_{j=1}^{k_1} w(I_j)=\sum_{j=1}^{k_1} \frac{\eps \cdot w(Z)} {\mu(Z) -(j-1)\eps}
$$
which can be estimated by  integrals,
$$ \int_{1}^{k_1} \frac{\eps \cdot w(Z)}{\mu(Z) - (j-1)\eps} dj~\leq~
\sum_{j=1}^{k_1} \frac{\eps \cdot w(Z)} {\mu(Z) -(j-1)\eps}
~\leq~  \int_{1}^{k_1+1} \frac{\eps \cdot w(Z)}{\mu(Z) - (j-1)\eps} dj.
$$
By solving the integrals we obtain,
\begin{equation}
\label{eq:II1_upper}
w(\II_1)\leq w(Z)\cdot \ln \left( \frac{\mu(Z)}{\mu(Z)-k_1\eps}\right)=
w(Z)\cdot \ln \left(\frac{\mu(Z) \cdot w(X) }{w(Z)\cdot \mu(X)}\right)\approx 0.0494,
\end{equation}
and 
\begin{equation}
	\label{eq:II1_lower}
w(\II_1)\geq w(Z)\cdot \ln \left( \frac{\mu(Z)}{\mu(Z)-(k_1-1)\eps}\right)
> 
w(Z)\cdot \ln \left(\frac{\mu(Z) \cdot w(X) }{w(Z)\cdot \mu(X)}\right) - \delta,
\end{equation}
where the last inequality follows from the definition of $k_1$ and \eqref{eq:eps_bound1}. 
Furthermore, we have that   $\mu(Z\setminus \mathcal{I}_1) =\mu(Z)-\eps k_1= \frac{\mu(X)}{w(X)} \cdot w(Z)$.

 Let $(I_j)_{j=k_1+1}^{k_2}$ be disjoint sets such that $I_j\subseteq X\cup Y\cup Z \setminus \II_1$, $\mu(I_j\cap X) = \eps \cdot w(X)$, $\mu(I_j\cap Y)= \eps \cdot w(Y)$ and $\mu(I_j \cap Z) =\eps \cdot w(Z)$ for $k_1< j \leq k_2$. Thus, $\mu(I_j)=\eps$ for $k_1<j\leq k_2$.
  It holds that 
  $$ \sum_{j=k_1+1}^{k_2} \eps \cdot w(Y)= \sum_{j=k_1+1}^{k_2} \eps \cdot w(X)= (k_2 -k_1 )\cdot \eps \cdot w(X) = \mu(X)-\rho\cdot w(X)\approx 0.144< \mu (X)=\mu(Y).$$
Similarly,
$$\sum_{j=k_1+1}^{k_2} \eps \cdot w(Z) =\eps\cdot w(Z) \cdot (k_2-k_1) = w(Z) \cdot \frac{\mu(X)}{w(X)}-\rho\cdot w(Z) = \mu(Z\setminus  \II_1) -\rho \cdot w(Z).$$
Hence, the sets $(I_j)_{j=k_1+1}^{k_2}$ exist. 

For any $k_1<j\leq k_2$ define $$w(I_j) = \frac{\eps \cdot w(Z)}{\aeps(Z)-k_1 \cdot\eps - w(Z)\cdot\eps\cdot (j-1-k_1)} = \frac{\eps \cdot w(X)}{\mu(X) - \eps \cdot w(X) \cdot (j-1-k_1)}.$$ Also, let $\II_2 = \bigcup_{j=k_1+1}^{k_2} I_j$.
 It holds that $w(\II_2) = 
 \sum_{j=k_1+1}^{k_2}  \frac{ \eps \cdot w(X)}{\mu(X) -\eps \cdot w(X)\cdot  (j-1-k_1)}$. Thus, 
$$ \int_{k_1+1}^{k_2} 
\frac{ \eps \cdot w(X)}{\mu(X)- \eps \cdot w(X)\cdot (j-1-k_1)}dj \leq w(\II_2) \leq 
\int_{k_1+1}^{k_2+1} 
\frac{ \eps \cdot w(X)}{\mu(X)- \eps \cdot w(X)  \cdot (j-1-k_1)}dj$$
and by solving the integrals and using \eqref{eq:eps_bound2} we have
\begin{equation} 
	\label{eq:II2_bounds}
 \ln \left( \frac{ \mu(X)}{ \rho\cdot  w(X) }\right)-\delta \leq \ln \left( \frac{ \mu(X)}{\rho\cdot w(X)+\eps \cdot w(X) }\right)\leq w(\II_2)\leq 
\ln \left( \frac{ \mu(X)}{ \rho\cdot  w(X) }\right).
\end{equation}
,


  We define $E= \{X,Y, Z, I_1,\ldots, I_{k_2}\}$. 
  	Consider the execution of Algorithm \ref{alg:greedy} on the MSK instance $(E,f,w,W)$ defined above.
  	Let $A_j= \{I_1,\ldots, I_j\}$ for $0\leq j \leq k_2$ ($A_0=\emptyset)$. 
  	For any $0\leq j \leq k_2$ and $j<\ell\leq k_2$ it holds that 
  	\begin{equation}
  		\label{eq:marginal_I}
  		\begin{aligned}
  	&\ell\leq k_1 &:~~~~~&\frac{f_{A_j}(I_\ell)}{w(I_{\ell})} = \frac{\mu(I_{\ell})}{w(I_{\ell})}  =
  	 \frac{\mu(Z)-\eps (\ell-1)}{w(Z)}\\
  	&k_1<\ell&:~~~~~&
  	\frac{f_{A_j}(I_\ell)}{w(I_{\ell})} = 
  	\frac{\mu(I_{\ell})}{w(I_{\ell})}=
  	 \frac{\mu(Z)- k_1\cdot \eps } {w(Z)} - \eps (\ell-1-k_1)= \frac{\mu(X)}{w(X) }- \eps (\ell-1-k_1)
  	\end{aligned}
  	\end{equation}
  	Moreover, for any $0\leq j < k_1$ it holds that,
  	\begin{equation}
  		\label{eq:marginals_k1}
  	\begin{aligned}
  		\frac{f_{A_j}(X)}{w(X)}= \frac{ \mu\left(X\setminus \left( \bigcup_{\ell=1}^{j} I_\ell \right) \right)}{w(X)} = \frac{\mu(X)}{w(X)}\approx 0.9368\\
  		\frac{f_{A_j}(Y)}{w(Y)}= \frac{ \mu\left(Y\setminus \left( \bigcup_{\ell=1}^{j} I_\ell \right) \right)}{w(Y)} = \frac{\mu(Y)}{w(Y)}\approx 0.9368\\
  		\frac{f_{A_j}(Z)}{w(Z)}= \frac{ \mu\left(Z\setminus \left( \bigcup_{\ell=1}^{j} I_\ell \right)\right)}{w(Z)} = \frac{\mu(Z) - \eps \cdot j}{w(Z)  }
  	\end{aligned}
  	\end{equation}
  By \eqref{eq:marginal_I}, \eqref{eq:marginals_k1} and the definition of $k_1$, it holds that for any $1\leq j \leq k_1$  $\max_{e\in E} \frac{f_{A_{j-1}}(\{e\})}{w(e)}= \frac{f_{A_{j-1}} (I_j)}{w(I_j)}$. 
  
  Similarly, for any $k_1\leq j \leq k_2$ it holds that,
  \begin{equation}
  	\label{eq:marginals_k2}
  	\begin{aligned}
  		\frac{f_{A_j}(X)}{w(X)}= \frac{ \mu\left(X\setminus \left( \bigcup_{\ell=1}^{j} I_\ell \right) \right)}{w(X)} = \frac{\mu(X)}{w(X)} - \eps\cdot ( j-k_1)\\
  		\frac{f_{A_j}(Y)}{w(Y)}= \frac{ \mu\left(Y\setminus \left( \bigcup_{\ell=1}^{j} I_\ell \right) \right)}{w(Y)} = \frac{\mu(Y)}{w(Y)} - \eps\cdot ( j-k_1)\\
  	  \frac{f_{A_j}(Z)}{w(Z)}= \frac{ \mu\left(Z\setminus \left( \bigcup_{\ell=1}^{j} I_\ell \right) \right)}{w(Z)} = \frac{\mu(X)}{w(X) }-\eps\cdot (j-k_1)
  	\end{aligned}
  \end{equation}
  
   By \eqref{eq:marginal_I} and~\eqref{eq:marginals_k2}, it follows that  $\max_{e\in E} \frac{f_{A_{j-1}}(\{e\})}{w(e)}= \frac{f_{A_{j-1}} (I_j)}{w(I_j)}$ for any  any $k_1< j \leq k_2$. Hence, we conclude that Algorithm \ref{alg:greedy} (when invoked from Step \ref{plus:greedy}) selects ${I_1,\ldots, I_{k_2}}$ in the first $k_2$ iterations, and subsequently selects $Z$. It holds that
   \begin{equation*}
   \begin{aligned}
   w&(\{I_1,\ldots, I_{k_2}, Z\})=
   w(\II_1)+w(\II_2) + w(Z) \\
   &\leq w(Z)+ w(Z)\cdot \ln \left(\frac{\mu(Z) \cdot w(X) }{w(Z)\cdot \mu(X)}\right)+ \ln \left( \frac{ \mu(X)}{ \rho\cdot  w(X) }\right)<1
   \end{aligned}
\end{equation*}
   where the first inequality is by \eqref{eq:II1_upper} and 
   \eqref{eq:II2_bounds}. Hence, the algorithm also adds the elements of $\{I_1,\ldots, I_{k_2}, Z\}$ to the solution in Step \ref{greedy:selected}.
   Denote $A= \{I_1,\ldots, I_{k_2}, Z\}$.
    It holds that,
   \begin{equation*}
   	\begin{aligned}
   		w(A) &=
   		w(\II_1)+w(\II_2)+w(Z) \\
   		&>
   		w(Z)+ 
   		w(Z)\cdot \ln \left(\mu(Z) \cdot \frac{\mu(Z)\cdot w(X) }{w(Z)\mu(X)}\right)
   		+
   	\ln \left( \frac{ \mu(X)}{ \rho\cdot  w(X) }\right)-2\delta \\
   		&>1-w(X)= 1-w(Y)
   	\end{aligned}
   \end{equation*}
   where the first inequality is due to  \eqref{eq:II1_lower} and \eqref{eq:II2_bounds}, and the second is a simple numerical inequality. Hence, Algorithm \ref{alg:greedy} cannot add any other element to the solution $A$ without violating the knapsack constraint. Thus, $A$ is indeed the solution output after the execution of Step \ref{plus:greedy} in Algorithm \ref{alg:greedy_plus}.
    
    It also holds that
    $$f(A)= \mu(Z) + \eps \cdot (k_2-k_1)(w(X)+w(Y)) =f_Z -2\rho W(X)+ 2f_X =1-2\rho \cdot w(X)< 0.42945=\beta.$$
    Also, $f(\{e\}) < \beta$ for any $e\in E$. Hence, Algorithm \ref{alg:greedy_plus} returns a set $S$ such that $f(S)<\beta$. It follows that the approximation ratio of $\greedyplus$ is bounded by $\beta$. 

\end{proof}
 
\section{Acknowledgments}
 Funding: This work has received funding from the European Union’s Horizon 2020 research and innovation program under grant agreement no. 852870-ERC-SUBMODULAR.

We would like to thank Moran Feldman for stimulating discussions and
for pointing us
to relevant literature.

\bibliographystyle{plain}
\bibliography{subknap}

\begin{thebibliography}{10}

\bibitem{BV14}
Ashwinkumar Badanidiyuru and Jan Vondr{\'a}k.
\newblock Fast algorithms for maximizing submodular functions.
\newblock In {\em Proceedings of the twenty-fifth annual ACM-SIAM symposium on
  Discrete algorithms (SODA)}, pages 1497--1514. SIAM, 2014.

\bibitem{cohen2008generalized}
Reuven Cohen and Liran Katzir.
\newblock The generalized maximum coverage problem.
\newblock {\em Information Processing Letters}, 108(1):15--22, 2008.

\bibitem{Co06}
Thomas~M. Cover and Joy~A. Thomas.
\newblock {\em Elements of Information Theory (Wiley Series in
  Telecommunications and Signal Processing)}.
\newblock Wiley-Interscience, New York, NY, USA, second edition, 2006.

\bibitem{EN19}
Alina Ene and Huy~L. Nguyen.
\newblock {A Nearly-Linear Time Algorithm for Submodular Maximization with a
  Knapsack Constraint}.
\newblock In {\em 46th International Colloquium on Automata, Languages, and
  Programming (ICALP)}, pages 53:1--53:12, 2019.

\bibitem{FKNRS20}
Yaron Fairstein, Ariel Kulik, Joseph~(Seffi) Naor, Danny Raz, and Hadas
  Shachnai.
\newblock A (1-e\({}^{\mbox{-1}}\)-{\(\epsilon\)})-approximation for the
  monotone submodular multiple knapsack problem.
\newblock In {\em 28th Annual European Symposium on Algorithms, {ESA} 2020,
  September 7-9, 2020, Pisa, Italy (Virtual Conference)}, volume 173 of {\em
  LIPIcs}, pages 44:1--44:19, 2020.

\bibitem{Fe98}
Uriel Feige.
\newblock A threshold of ln n for approximating set cover.
\newblock {\em J. ACM}, 45(4):634–652, July 1998.

\bibitem{FNS20}
Moran Feldman, Zeev Nutov, and Elad Shoham.
\newblock Practical budgeted submodular maximization.
\newblock {\em arXiv preprint arXiv:2007.04937}, 2020.

\bibitem{KMS99}
Samir Khuller, Anna Moss, and Joseph Naor.
\newblock The budgeted maximum coverage problem.
\newblock {\em Information processing letters}, 70(1):39--45, 1999.

\bibitem{LKGFVG07}
Jure Leskovec, Andreas Krause, Carlos Guestrin, Christos Faloutsos, Jeanne
  VanBriesen, and Natalie Glance.
\newblock Cost-effective outbreak detection in networks.
\newblock In {\em Proceedings of the 13th ACM SIGKDD international conference
  on Knowledge discovery and data mining}, pages 420--429, 2007.

\bibitem{LB10}
Hui Lin and Jeff Bilmes.
\newblock Multi-document summarization via budgeted maximization of submodular
  functions.
\newblock In {\em Human Language Technologies: The 2010 Annual Conference of
  the North American Chapter of the Association for Computational Linguistics},
  pages 912--920, 2010.

\bibitem{Mc86}
Alex McNabb.
\newblock Comparison theorems for differential equations.
\newblock {\em Journal of mathematical analysis and applications},
  119(1-2):417--428, 1986.

\bibitem{NW78}
G.~L. Nemhauser and L.~A. Wolsey.
\newblock Best algorithms for approximating the maximum of a submodular set
  function.
\newblock {\em Mathematics of Operations Research}, 3(3):177--188, 1978.

\bibitem{NWF78}
George~L Nemhauser, Laurence~A Wolsey, and Marshall~L Fisher.
\newblock An analysis of approximations for maximizing submodular set
  functions—i.
\newblock {\em Mathematical programming}, 14(1):265--294, 1978.

\bibitem{Si13}
Thomas~C Sideris.
\newblock {\em Ordinary differential equations and dynamical systems}.
\newblock Springer, 2013.

\bibitem{SKYK11}
K.~{Son}, H.~{Kim}, Y.~{Yi}, and B.~{Krishnamachari}.
\newblock Base station operation and user association mechanisms for
  energy-delay tradeoffs in green cellular networks.
\newblock {\em IEEE Journal on Selected Areas in Communications},
  29(8):1525--1536, 2011.

\bibitem{Sv04}
Maxim Sviridenko.
\newblock A note on maximizing a submodular set function subject to a knapsack
  constraint.
\newblock {\em Operations Research Letters}, 32(1):41--43, 2004.

\bibitem{tang2021revisiting}
Jing Tang, Xueyan Tang, Andrew Lim, Kai Han, Chongshou Li, and Junsong Yuan.
\newblock Revisiting modified greedy algorithm for monotone submodular
  maximization with a knapsack constraint.
\newblock {\em Proceedings of the ACM on Measurement and Analysis of Computing
  Systems}, 5(1):1--22, 2021.

\bibitem{yaroslavtsev2020bring}
Grigory Yaroslavtsev, Samson Zhou, and Dmitrii Avdiukhin.
\newblock “bring your own greedy”+ max: Near-optimal 1/2-approximations for
  submodular knapsack.
\newblock In {\em International Conference on Artificial Intelligence and
  Statistics}, pages 3263--3274. PMLR, 2020.

\end{thebibliography}
\appendix
\section{Proof of Lemma \ref{lem:comparison}}
\label{sec:app}

To prove Lemma \ref{lem:comparison} we first show a simpler claim.
\begin{lemma}
	\label{lem:comparison_aux}
Within the settings of Lemma \ref{lem:comparison}, if $ \vartheta_1(c_r)\geq \vartheta_2(c_r) $ for some $r\in [s]$ then $\vartheta_1(u)\geq \vartheta_2(u)$ for any $u\in [c_r,c_{r+1}]$. 
\end{lemma}
\begin{proof}
If $c_r =c_{r+1}$ the claim trivially holds. Therefore we may assume that $c_r<c_{r+1}$. 
Define $\delta:[c_r, c_{r+1}]\rightarrow \mathbb{R}$ by $\delta(u)= \vartheta_2(u)-\vartheta_1(u)$, and let $\delta'=\vartheta_2'-\vartheta_1'$ be its derivative. We note that $\vartheta_2'$ and $\vartheta_1'$ are continuous in $(c_r, c_{r+1} )$ and hence $\delta'$ is continuous and integrable on   $(c_r, c_{r+1} )$. Thus, for any $u\in (c_r, c_{r+1})$ it holds that
\begin{equation}
\label{eq:integral}
\delta(u)=  \lim_{t\searrow c_r}\delta(u)-\delta(t) +\delta(c_r) = \lim_{t\searrow c_r} \int_{t}^{u}\delta'(z)dz +\delta(c_r) =
\delta(c_r)+\int_{c_r}^{u} \delta'(z)dz
\end{equation}
where the first equality holds since $\delta$ is continuous.
Furthermore, as $\varphi_r$ is positively linear in the second dimension there is $K_r>0$ such that $\varphi_r(u,t_1)-\varphi_r(u,t_2)= K_r \cdot  (t_1-t_2)$ for any $u\in (c_r,c_{r+1})$ and $t_1, t_2\in \mathbb{R}$. 
Hence, for any $u\in (c_r,c_{r+1})$, it holds that
\begin{equation*}
\begin{aligned}
\delta'(u) &=\vartheta_2'(u)-\vartheta_1'(u) \leq \varphi_r(u,\vartheta_2(u))- \varphi_r(u, \vartheta_1(u)) = K_r \cdot \left(
\vartheta_2(u)-\vartheta_1(u) \right)\\
& = K_r \left( \delta(c_r)+ \int_{c_r}^{u} \delta'(z)dz\right)\leq
K_r\cdot {\left|\int_{c_r}^{u} \delta'(z)dz\right|}.
\end{aligned}
\end{equation*}
The last equality follows from \eqref{eq:integral} and the  inequality holds since $\delta(c_r) = \vartheta_2(c_r)-\vartheta_1(c_r)\leq 0$.
By Gronwall's inequality (see, e.g., Lemma 3.3 in \cite{Si13}), it follows that $\delta(u)\leq 0$ for any $u\in (c_r, c_{r+1})$, and therefore $\vartheta_1(u)\geq \vartheta_2(u)$ for any $u\in(c_r, c_{r+1})$. As $\delta$ is continuous,
$
 \delta(c_{r+1}) =
\lim_{u\nearrow c_{r+1} } \delta(u) \leq 0
$; thus, $\vartheta_1(c_{r+1})\geq \vartheta_2(c_{r+1})$ as well.
\end{proof}


\begin{proof}[Proof of Lemma \ref{lem:comparison}]
	
The lemma essentially follows immediately form Lemma \ref{lem:comparison_aux} using an inductive claim. We will prove by induction on $r\in [s+1]$ that $\vartheta_1(u)\geq \vartheta_2(u)$ for any $u\in [a,c_r]$. For $r=1$ the claim holds since $c_1=a$ and $\vartheta_1(a)\geq \vartheta_2(a)$. 
Let $r>1$ and assume $\vartheta_1(u)\geq \vartheta_2(u)$ for any $u\in [a,c_r] $. Then, by Lemma \ref{lem:comparison_aux}, $\vartheta_1(u)\geq \vartheta_2(u)$ for any $u\in [c_r,c_{r+1}]$ as well. That is, the claim holds for $r+1$. 

Taking $r=s+1$ we have that $\vartheta_1(u)\geq \vartheta_2(u)$ for any $u\in [a,c_{s+1}]=[a,b]$
\end{proof}
\end{document}